\documentclass[12pt]{article}

\usepackage{graphicx}
\usepackage{amssymb,amsmath,color}
\usepackage[breaklinks=true]{hyperref}
\usepackage[anythingbreaks]{breakurl}
\usepackage[round]{natbib}
\usepackage{amsthm}
\newtheorem{thm}{Theorem}

\newcommand{\beq}{\begin{equation}}
\newcommand{\eeq}{\end{equation}}
\newcommand{\ben}{\begin{enumerate}}
\newcommand{\een}{\end{enumerate}}

\newcommand{\comment}[1]{}

\title{Next Steps for the Colorado Risk-Limiting Audit (CORLA) Program}

\author{
   Mark Lindeman\\
   Neal McBurnett\\
   Kellie Ottoboni\\
   Philip B.~Stark
}

\date{Version: \today}

\begin{document}
\maketitle

\begin{abstract}
Colorado conducted risk-limiting tabulation audits (RLAs) across the state in 2017,
including both ballot-level comparison audits and ballot-polling audits.
Those audits only covered contests restricted to a single county;
methods to efficiently audit contests that cross county boundaries
and combine ballot polling and ballot-level comparisons have not been available.

Colorado's current audit software (RLATool) needs to be improved to audit
these contests that cross county lines and to audit small contests efficiently.

This paper addresses these needs. 
It presents extremely simple but inefficient methods, more efficient methods
that combine ballot polling and ballot-level comparisons using stratified samples,
and methods that combine ballot-level comparison and
variable-size batch comparison audits in a way that does not require stratified
sampling.

We conclude with some recommendations, and illustrate our recommended method
using examples that compare them to existing approaches.
Exemplar open-source code and interactive Jupyter notebooks are provided
that implement the methods and allow further exploration.
\end{abstract}

\noindent
\textbf{Acknowledgements.}
We are grateful to Ronald L.~Rivest for helpful conversations and suggestions.

\section{Introduction}
A risk-limiting audit (RLA) of an election is a procedure that
has a known, pre-specified minimum chance of correcting the electoral outcome if the outcome
is incorrect---that is, if the reported outcome differs from the outcome that a full manual
tabulation of the votes would find. 
RLAs require a durable, voter-verifiable record of voter intent, such as paper ballots,
and they assume that this audit trail is sufficiently complete and accurate that a full hand
tally would show the true electoral outcome.
That assumption is not automatically satisfied: a \emph{compliance audit}
\citep{starkWagner12} 
is required.

Risk-limiting audits are generally incremental: they examine more ballots, or batches of ballots,
until either (i)~there is strong statistical evidence that a full hand tabulation would confirm the outcome,
or (ii)~the audit has led to a full hand tabulation, the result of which should become the official
result.

RLAs have been piloted in California, Colorado, and Ohio, and a test of
RLA procedures has been conducted in Arizona.
RLA bills are being drafted or are already under consideration in California,
Virginia, Washington, and other states.
A number of laws have either allowed or mandated risk-limiting audits,
including California AB~2023 (Salda\~{n}a), SB~360 (Padilla), and AB~44 (Mullin);
Rhode Island SB~413A and HB~5704A; and Colorado Revised Statutes (CRS)~1-7-515.

CRS~1-7-515 required 
Colorado to implement risk-limiting audits beginning in 2017.
(There are provisions to allow the Secretary of State to exempt some counties.)
The first set of coordinated risk-limiting election audits across the state took place in Colorado in November, 2017.\footnote{%
 See \url{https://www.sos.state.co.us/pubs/elections/RLA/2017RLABackground.html}
}

Colorado's ``uniform voting system'' program\footnote{%
\url{https://www.sos.state.co.us/pubs/elections/VotingSystems/UniformVotingSystem.html}
} led
many Colorado counties to purchase (or to plan to purchase) voting systems
that are auditable at the ballot level: those systems export cast vote records (CVRs)
for individual ballots in a manner that allows the corresponding paper ballot to be identified,
and conversely, make it possible to find the CVR corresponding to any
particular paper ballot.
We call counties that have such systems ``CVR'' counties.
It is estimated that by June, 2018, 98.2\% of active Colorado voters will be in CVR counties.
CVR counties can perform ``ballot-level comparison audits,'' \citep{lindemanStark12} 
which are currently the
most efficient approach to risk-limiting audits in that they require examining
fewer ballots than other methods do, when the outcome of the contest under audit 
is in fact correct.

Other counties (``legacy'' or ``no-CVR'' counties) 
have systems that do not allow auditors to check how the system
interpreted voter intent for individual ballots.
Their election results can still be audited, provided their voting systems
create a voter-verifiable paper trail (\emph{e.g.}, voter-marked paper ballots) that is
conserved to ensure that it remains accurate and intact, and organized well enough
to permit ballots to be selected at random.
Pilot audits in California suggest that the most efficient way to audit such systems
is by ``ballot-polling'' \citep{lindemanEtal12,lindemanStark12} 
(in contrast to ``batch-level comparisons,'' for example).

There is currently no literature on how to perform risk-limiting audits 
of contests that include CVR counties and no-CVR counties by combining
ballot polling and ballot-level comparisons.
Existing methods would either require all counties to use the lowest
common denominator, ballot-polling (which does not take advantage of the CVRs,
and thus is expected to require more auditing than a method that does take
advantage of the CVRs), or would
require no-CVR counties to perform batch-level comparisons, which were found in
California to be (generally) less efficient than ballot-polling audits.%
\footnote{%
  See~\cite{Rivest-2018-bayesian-tabulation-audits}
  for a different (Bayesian) approach to auditing contests that include both CVR counties
  and no-CVR counties.
  }

The open-source audit software used for Colorado's 2017 audits,
RLATool (\url{https://github.com/FreeAndFair/ColoradoRLA/}),
needs to be improved to audit
contests that cross county lines and to audit small contests efficiently.

First, the current version (1.1.0) of RLATool
needs to be modified to recognize and group together contests that cross jurisdictional
boundaries; currently, it treats every contest as if it were entirely
contained in a single county.
Margins and risk limits apply to entire contests, not to the portion of a contest
included in a county.
RLATool also does not allow the user to select the sample size, nor does
it directly allow an unstratified random sample to be drawn across counties.
Second, to audit a contest that includes voters in ``legacy'' counties
(counties with voting systems that cannot export cast vote records)
and voters in counties with newer systems, new statistical methods are needed to
keep the efficiency of ballot-level comparison audits that the newer systems
afford.
Third, auditing contests that appear only on a subset of ballots can
be made much more efficient if the sample can be drawn from just those ballots
that contain the contest.
While allowing samples to be restricted to ballots reported to contain a particular
contest is not essential in the short run, it will be necessary
eventually to make it feasible to audit smaller contests.

This document focuses on near-term requirements for risk-limiting audits in Colorado.
Section~\ref{sec:crude} presents a number of crude approaches
that could be implemented easily
but might require examining substantially more ballots.
Section~\ref{sec:variable}
presents an approach based on comparison audits with different batch sizes.
This approach is statistically elegant and relatively efficient, but might require changing how counties
handle their ballots.
Section~\ref{sec:hybrid} presents our recommended approach, which combines ballot-level
comparisons in counties that can perform them with ballot-polling in the no-CVR counties.
All the approaches require new software, including at least minor modifications to RLATool.
We provide example software implementing the risk calculations for
our recommended approach as a Python Jupyter notebook.\footnote{%
 See \url{https://github.com/pbstark/CORLA18}.
}
Section~\ref{sec:subcollections}
describes how audit efficiency could be improved in CVR counties by combining CVR
data with data from Colorado's voter registration system, SCORE.\footnote{%
  SCORE is Colorado's voter registration system, which also tracks who voted.
  See \url{https://www.sos.state.co.us/pubs/elections/SCORE/SCOREhome.html}.
} 
Sections~\ref{sec:comparisonError} and~\ref{sec:ballotPollError} explain the recommended
modifications to ballot-level comparison and ballot-polling audits, respectively. 
Section~\ref{sec:recommendations} summarizes our recommendations and
considerations for implementation.

\subsection{Priorities for Colorado audits}

Auditing efficiency is controlled in part by how well the audit can limit the sample to ballots that
contain the contests under audit.
Some contests are on (essentially) every ballot, for instance the governor's race.
Others, such as mayoral contests, may appear on only a small fraction of ballots cast in
a county.
Partisan primaries---even for statewide office---are somewhere in between,
because in general no single party's primary appears on every ballot cast in the state.
Thus, either we accept reduced 
efficiency for the sake of simplicity by continuing to sample ballots uniformly from within 
counties (or collections of counties), or we develop a way to
focus the auditing on the ballots that contain the contest.
The latter requires external information, e.g., from SCORE,
as discussed below.

Moreover, party primaries for statewide offices (and perhaps other contests) will
include CVR counties and no-CVR counties, so we need a method to audit
across both kinds of voting technology. 

This report addresses both issues, providing 
options for effectively auditing  
heterogeneous 
voting technology, varying in efficiency, complexity, and on whom any additional audit burden falls.

\section{Crude (and unpleasant) approaches} \label{sec:crude}


Here and generally throughout the paper, we 
discuss auditing a single contest at a time, although the same sample can be used to audit
more than one contest and there are ways of combining audits of different contests into
a single process \citep{stark09c,stark10d}.
We use terminology drawn from a number of papers; the key reference is \citet{lindemanStark12}.
An \emph{overstatement error} is an error that caused the margin between \emph{any} reported
winner and \emph{any} reported loser to appear larger than it really was.
An \emph{understatement error} is an error that caused the margin between \emph{every} reported
winner and \emph{every} reported loser to appear to be smaller than it really was.

\subsection{Hand count the legacy counties}
The simplest approach to combining legacy counties with CVR counties is to require every
legacy county to do a full hand count of the primaries, and to conduct a 
ballot-level comparison audit in CVR counties, based on contest margins adjusted for
the results of the manual tallies in the CVR counties.
For instance, imagine a contest with two candidates, reported winner $w$ and reported loser $\ell$.
Suppose the total number of reported votes for candidate~$w$ is $V_w$ 
and the total for candidate~$\ell$ is $V_\ell$, so that $V_w > V_\ell$, since 
$w$ is the reported winner.
Suppose that a full manual tally of the votes in the legacy counties shows $V_w'$ votes for $w$ and
$V_\ell'$ votes for $\ell$.
Suppose that a total of $N$ ballots were cast in the CVR counties.
Then the \emph{diluted margin}  for the comparison audit in the CVR counties is defined to be
$[(V_w-V_w')-(V_\ell-V_\ell')]/N$.
Requiring a full hand count in the legacy counties has obvious disadvantages, except perhaps in very 
close contests where ballot polling is not efficient. (But it does have the advantage of not forcing CVR 
counties to do additional auditing to compensate for the legacy counties.)

\subsection{Subtract error bounds for the legacy counties from vote totals}
If ballot accounting and SCORE data can provide good upper bounds on the number of ballots cast in
each contest in legacy counties, there are simple upper bounds on the total
possible overstatement error each legacy county could contribute to the overall contest
results; those can be subtracted from the overall margin (as in the previous subsection) and the
remainder of the contests can be audited in CVR counties against the adjusted margins.
For instance, consider a primary that appears on $N$ ballots in a legacy counties.
Suppose that in legacy counties, the overall, statewide contest winner, $w$, is reported to have received $V_w'$ votes, and some loser, $\ell$, is reported to have received $V_\ell'$ votes. 
(Note that $V_\ell'$ could be greater than $V_w'$: $w$ is not necessarily the reported winner in the legacy counties.)
Then the most overstatement error that the county could possibly have in determining whether
$w$ in fact beat $\ell$ is if every reported undervote, invalid vote, or vote for a different candidate, $t$, had 
in fact been a vote for $\ell$ (producing a 1-vote overstatement), and every vote reported for 
$w$ was in fact a vote for $\ell$ (producing a 2-vote overstatement).
The reduction in the margin that would produce is 
$N - V_w' - V_\ell' + 2V_w' = N + V_w' - V_\ell'$ votes.

Whereas the previous approach places the auditing burden created by obsolescent equipment entirely on 
the legacy counties, this approach places it entirely upon the CVR counties. Also, in a close contest, it 
could require a full hand count in every county that might not otherwise be necessary.

\subsection{Treat legacy counties as if every ballot selected from them for audit has a two-vote overstatement}\label{sec:two-vote-over}
A third simple-but-pessimistic approach is to sample uniformly from all counties as if one
were performing a ballot-level comparison audit everywhere,  but to 
treat any ballot
selected from a legacy county as a two-vote overstatement. This approach has the same disadvantages
as the previous approach.

\section{Variable batch sizes} \label{sec:variable}

Another approach is to perform a comparison audit across all counties, but to use batches consisting
of more than one ballot (batch-level comparisons)
in legacy counties and batches consisting of a single ballot (ballot-level comparisons) in CVR counties.\footnote{%
 For majority and plurality elections, including those in which voters can select more than one candidate,
  audits can be based on overstatement and understatement errors at the level of batches.
}
This requires that the no-CVR counties report vote subtotals
for physically identifiable batches.
If a county's voting system can only report subtotals by precinct but 
the county does not sort paper ballots by
precinct, this approach might require revising how the county handles its
paper; we understand that this is the case in many Colorado counties.

That said, many California counties that do not sort vote-by-mail (VBM)
ballots by precinct conduct the statutory 1\% audits by manually retrieving the ballots 
for just those precincts selected for audit from whatever physical batches they happen to be in: 
the situation is identical to that in Colorado.

Another solution is the ``Boulder-style'' batch-level audit,\footnote{%
 See \url{http://bcn.boulder.co.us/~neal/elections/boulder-audit-10-11/}.
}
which requires generating 
vote subtotals after each physical batch is scanned, and exporting those subtotals in machine-readable form.
That in turn may require using extra memory cards, repeatedly initializing and deleting tabulation databases,
or other measures that add complexity and opportunity for human error.

While those two approaches are laborious, they would provide a viable short-term solution,
especially combined with information from SCORE to check that the reported batch-level results contain the correct number of ballots for each contest under audit.
Moreover, it does not unduly increase the workload in CVR counties
to compensate for legacy equipment.

This kind of variable-batch-size comparison audit approach would require modifying or augmenting
RLATool in several ways: 

\begin{enumerate}

  \item The CVR reporting tool would need to be modified to allow no-CVR counties to
report batch-level results in a manner analogous to how CVR counties report
ballot-level results, or an external tool would need to be provided.

  \item The sampling
algorithm would have to allow sampling batches---and sampling them with unequal probability,
because efficient batch-level audits involve sampling batches with probability proportional
to a bound on the possible overstatement error in the batch.
It would also need to calculate the appropriate sampling probability for each batch (of whatever size).
Again, this could be accommodated using an external tool to draw the sample from legacy counties.

  \item The risk calculations would need to be modified. 
This, too, could be done with external software, with suitable provisions for capturing audit data
from RLATool or directly from legacy counties.
\end{enumerate}

None of these changes is enormous; the mathematics and statistics are already worked out
in published papers, and there is exemplar code for calculating the
batch-level error bounds, drawing the samples with probability proportional to an
error bound, and calculating the attained risk from the sample results.
Indeed, this is the method that was used in several of California's pilot audits,
including the audit in Orange County.
A derivation of a method for comparison audits with variable batch sizes is given below
in section~\ref{sec:comparisonError}.

\section{Stratified ``hybrid'' audits} \label{sec:hybrid}

Other approaches involve \emph{stratification}: partitioning the cast ballots
into non-overlapping groups and sampling independently from those groups.
One could stratify by county, but in general it is simpler and more efficient
statistically (i.e., results in auditing fewer ballots) to minimize the number of strata.
We consider methods that use two strata: CVR counties and no-CVR counties. 
Collectively, the ballots cast in CVR counties comprise one stratum and the ballots cast in 
legacy counties comprise a second stratum; every ballot cast in the contest is in 
exactly one of the two strata. 
We assume that the samples are drawn from the
two strata independently.

As explained below, these stratified ``hybrid'' audits require the specification of some additional parameters:
$\lambda_1$ for dividing the tolerable overstatement error up, and the strata risk limits $\{\alpha_s\}$.

\subsection{Partitioning the total permissible overstatement into strata}
The simplest approach to stratification involves partitioning the risk limit and the tolerable
overstatement error of the tabulation into
two pieces, one for the (pooled) CVR counties and one for the (pooled) no-CVR counties.
Let $V_{w\ell} > 0$ denote the contest-wide margin (in votes) of reported winner 
$w$ over reported loser
$\ell$.
Let $V_{w\ell,s}$ denote the margin (in votes) of reported winner $w$ over reported loser $\ell$
in stratum $s$. 
Note that $V_{w\ell,s}$ might be negative in one stratum.
Let $A_{w\ell}$ denote the margin (in votes)
of reported winner $w$ over reported loser $\ell$ that 
a full hand count of the entire contest would show, that is, the \emph{actual} margin rather
than the \emph{reported} margin.
Reported winner $w$ really beat reported loser $\ell$ if and only if $A_{w\ell} > 0$.
Define $A_{w\ell,s}$ to be the actual margin (in votes) of $w$ over $\ell$ in stratum $s$;
this too may be negative.

Let $\omega_{w\ell,s} \equiv V_{w\ell,s} - A_{w\ell,s}$ be the \emph{overstatement}
of the margin of $w$ over $\ell$ in stratum $s$.
Reported winner $w$ really beat reported loser 
$\ell$ if and only if $\omega_{w\ell} \equiv \omega_{w\ell,1} + \omega_{w\ell,2} < V_{w\ell}$.

Pick $\lambda_1 \in \Re$ and define $\lambda_2 = 1-\lambda_1$. These values partition the
total tolerable overstatement between the two strata:
If $\omega_{w\ell,1} < \lambda_1 V_{w\ell}$ \emph{and} 
$\omega_{w\ell,2} < \lambda_2 V_{w\ell}$, candidate $w$ really received more votes
than candidate $\ell$.
Some $(\lambda_1, \lambda_2)$ pairs 
can be ruled out \emph{a priori}, because (for instance) $\omega_{w\ell,s} \in [-2N_s, 2N_s]$,
where $N_s$ is the number of ballots cast in stratum $s$.
There are other simple, sharper bounds, sketched below.

The choice of $\lambda_1$ (which determines the tolerable overstatement in each
stratum), the strata risk limits $\{\alpha_s\}$, and details of the
audit procedures affect the workload and the overall risk limit.
(See section~\ref{sec:stratumRisk} and section~\ref{sec:recommendations}.)

For ballot-level comparison audits, auditing to ensure that $\omega_{w\ell,s} < \lambda_s V_{w\ell}$
is discussed in section~\ref{sec:comparisonError}.  It is a minor modification of the method
embodied in RLATool.

For ballot-polling audits, auditing to ensure that $\omega_{w\ell,s} < \lambda_s V_{w\ell}$ is discussed in section~\ref{sec:ballotPollError}.
Note that this requires a more substantial modification of the standard ballot-polling calculations,
because the standard calculations consider only the fraction of ballots with a vote for either 
$w$ or $\ell$ that contain a vote for $w$, while we need to make an inference about the 
difference between the number of votes for $w$ and the number of votes for $\ell$.
This introduces an additional unknown nuisance parameter, the number of ballots with votes for either
$w$ or $\ell$.

\subsubsection{Combining stratum-level risk limits}\label{sec:stratumRisk}
We audit to test the two hypotheses $\{\omega_{w\ell,s} \ge \lambda_s V_{w\ell}\}_{s=1}^2$, 
independently for the two strata.
If we reject \emph{both} hypotheses, we conclude that the contest outcome is correct;
otherwise, we manually re-tabulate the contest in one or both strata, depending on the
audit rules.
Those rules matter:
the two audits might need to be conducted to smaller risk limits individually than the desired
risk limit for the contest as a whole.

Recall that the samples are drawn independently from the two strata.
Pick $\alpha_1, \alpha_2 \in (0,\alpha)$.
(Below we discuss the choice further.)
We audit each stratum $s$ to test the hypothesis $\omega_{w\ell,s} \ge \lambda_s V_{w\ell}$ 
(the overstatement exceeds the tolerable overstatement) 
at risk limit $\alpha_s$,
as if it were its own election.
The audits can be conducted at the same time or sequentially; there is no coordination
between the audits unless one of them leads to a full hand count but the other does not:
see below.

How do these two stratum-level ``risk limits'' $\alpha_1$ and
$\alpha_2$ determine the 
overall risk that the audit will not correct the outcome if the outcome is wrong?
The overall risk depends on the rule for what we do if the audit in one stratum leads 
to a full manual tally of that stratum.

Here are the possibilities. Bear in mind that for the outcome to be wrong, 
at least one stratum must have a net overstatement
greater its tolerable overstatement: 
That is, if $\omega_{w\ell,1} + \omega_{w\ell,2} \ge V_{w\ell}$, then $\omega_{w\ell,1}\ge \lambda_1V_{w\ell}$
or $\omega_{w\ell,2}\ge \lambda_2V_{w\ell}$, or both.
If the tolerable overstatement is exceeded in only one stratum, $h$, then the chance that the 
stratum will be fully hand counted is at least $1-\alpha_h \ge 1- \alpha$.

If both $\omega_{w\ell,1} \ge \lambda_1V_{w\ell}$
and $\omega_{w\ell,2} \ge \lambda_2V_{w\ell}$, then the chance both 
are completely tabulated by hand is at least
$(1-\alpha_1)(1-\alpha_2)$, since the audit samples in the two strata are independent.

What should we do if the audit leads to a full tally in one stratum, $h$,
that reveals that indeed its tolerable overstatement has been exceeded, 
but the other audit has not led to a full tabulation, because it 
has not started, because it is still underway, or because it terminated without
a full hand tally?
We consider two options.
The simpler is to automatically require a full hand count of the other stratum. 
If the audit uses this rule, then we can take $\alpha_1 = \alpha_2 = \alpha$, 
and the procedure will have risk limit~$\alpha$. However, this rule creates the
possibility of requiring a full hand count in circumstances where it may seem
substantively superfluous. For instance, one can imagine an audit of a statewide
contest in which the tolerable overstatement in no-CVR counties is exceeded, 
yet the outcome still could be verified without a full hand count in the CVR counties.

The second approach is to adjust the tolerable overstatement in the other
stratum in light of the known manual tally $A_{w\ell,h}$
in the stratum $h$ that has been fully hand tallied:
we will test against the threshold 
$V_{w\ell} - A_{w\ell,h} \equiv \lambda_t' V_{w\ell}$, rather than 
the original value $\lambda_t V_{w\ell}$. (Because the overstatement
in stratum $h$ exceeded the tolerable overstatement, the updated tolerable
overstatement in stratum $t$ will be smaller than the original value.) 
Then to reject the new null hypothesis in stratum $t$ is to conclude that the 
overall outcome is correct.

If and when the hypothesis in stratum $t$ changes, the audit
in that stratum might be able to stop on the basis of the data already observed;
it might need to continue; or---if it had stopped based on the original threshold
$\lambda_t V_{w\ell}$---it might need to examine more ballots, possibly
continuing to a full hand tally.

We will now show in detail that this rule allows the contest to be audited at 
risk limit~$\alpha$ by selecting values of~$\alpha_1$ and~$\alpha_2$ that sum to
a bit more than~$\alpha$: specifically, such that $(1-\alpha_1)(1-\alpha_2) < 1-\alpha$.
For instance, suppose we want the overall risk limit to be 5\%. 
If we use a risk limit of 4\% in the no-CVR stratum and a risk limit of 1.04\% in the CVR stratum,
the overall risk limit is not larger than $1 - (1-\alpha_1)(1-\alpha_2) \equiv 1 - 0.96\times 0.9896 < 0.05$.

The statistical wrinkle is that adjusting for the manual tally in the hand-counted 
stratum $h$
changes the hypothesis being tested in the other stratum $t$
in a way that is itself random:
whether the original null $\omega_{w\ell,s} \ge \lambda_t V_{w\ell}$ is tested
or the new null $\omega_{w\ell,s} \ge \lambda_t' V_{w\ell}$ is tested depends on what the 
sample reveals in stratum $h$.
If the hypothesis does change, there is only one value possible for $\lambda_t'$---which
depends on the reported margin $V_{w\ell}$ and the count $A_{w\ell,h}$ in 
stratum $h$---but $\lambda_t'$ is unknown until $A_{w\ell,h}$ is known.

We assume that before any data are collected, the audit specifies two families of tests:
for each stratum $s$, a family of level-$\alpha_s$ tests of the null hypothesis that 
the overstatement in the stratum is greater than or equal to $c$, for all feasible values of $c$.
That is,
\beq
    \Pr \{ \mbox{reject hypothesis that } \omega_{w\ell,s} \ge 
    c_s || \omega_{w\ell,s} \ge c_s \} \le \alpha_s,
\eeq
for $s = 1, 2$, and all feasible $c_s$.
Moreover, we insist that the test depend on data only from ballots selected from its stratum.
Because the samples in the two strata are independent, for all feasible pairs $c_1, c_2$,
\begin{align} \label{eq:stratum_families}
    \Pr\{&\mbox{reject neither hypothesis } \omega_{w\ell,s} \ge c_s, \;\; s=1, 2 ||
       \omega_{w\ell,s} \ge c_s  \mbox{ for both } s=1, 2 \} \nonumber \\ 
       &= \prod_{s=1}^2 1 - \Pr \{ \mbox{reject hypothesis that } \omega_{w\ell,s} \ge c_s || \omega_{w\ell,s} \ge c_s \} \nonumber \\
       & \ge (1-\alpha_1)(1-\alpha_2).
\end{align}

What is the chance that the audit leads to a full hand tabulation if the outcome is incorrect?
One way the audit can lead to a full hand tally is if it leads to a full count in one stratum, 
the null hypothesis in the other stratum is changed, and the audit in the second 
stratum then proceeds to a full manual tally.
(There are other ways the audit can lead to a full hand tally, for instance, if neither
null hypothesis is rejected, but this is one way.)

If the outcome is wrong, there is at least one stratum in which the overstatement 
$\omega_{w\ell,s}$ 
exceeds the threshold $\lambda_s V_{w\ell}$.
Let $h$ be one such stratum. 
Then the chance the audit in stratum $h$ leads to a full manual tally in that stratum
is at least $(1-\alpha_h)$.
If the audit leads to a full manual tally in stratum~$h$ and the overall outcome is wrong,
then the (new) null hypothesis in the other stratum, $t$, must be true.
If we started to audit that new hypothesis \emph{ab initio}, the chance that we would reject it
would be at most $\alpha_t$, so the chance the audit would lead to a full hand count 
of stratum $t$ is at least $1-\alpha_t$.
The question is whether ``changing hypotheses'' could make that chance smaller.
The inequality \ref{eq:stratum_families} shows that it cannot: for any feasible pair of
overstatements, $c = (c_1, c_2)$, if $\omega_{w\ell,1} \ge c_1$ and $\omega_{w\ell,2} \ge c_2$,
the chance that neither the hypothesis $\omega_{w\ell,1} \ge c_1$ nor the hypothesis 
$\omega_{w\ell,2} \ge c_2$ will be rejected is at least $(1-\alpha_1)(1-\alpha_2)$.

And therefore, for this procedure, the chance that there will be a full hand count in both strata is at least 
$(1-\alpha_1)(1-\alpha_2)$ if the outcome is incorrect,
even if the probability were zero that both of the original audits would proceed to a full hand count.
The overall risk limit is thus not larger than $1 - (1-\alpha_1)(1-\alpha_2)$
.

\subsection{Constraining the total overstatement across strata}
A more statistically efficient approach to ensuring that the overstatement error in the 
two strata does not
exceed the margin is to try to constrain the \emph{sum} of the overstatement errors in the two
strata, rather than constrain the pieces separately:
there are many ways that the total overstatement could be less than $V_{w\ell}$ without
having the overstatement $\omega_{w\ell,s}$
in stratum $s$ less than $\lambda_s V_{w\ell}$, $s = 1, 2$.
To that end, imagine \emph{all} values $\lambda_1$.
If, for all such pairs, we can reject the hypothesis that the 
overstatement error in stratum~1 is greater than or equal to $\lambda_1 V_{w\ell}$ \emph{and} 
the overstatement error in stratum~2 is greater than or equal to $\lambda_2 V_{w\ell}$, then
we can conclude that the outcome is correct.

To test the conjunction hypothesis (i.e., that both of those null hypotheses are false), we use 
Fisher's combining function.
Let $p_s(\lambda)$ be the $p$-value of the hypothesis $\omega_{w\ell,s} \ge \lambda V_{w\ell}$.
If the null hypothesis that $\omega_{w\ell,1} \ge \lambda_1 V_{w\ell}$ and 
$\omega_{w\ell,2} \ge \lambda_2 V_{w\ell}$ is true, then the combination
\beq
   \chi(\lambda_1, \lambda_2) = -2 \sum_{s=1}^2 \ln p_s(\lambda_s)
\eeq
has a probability distribution that is dominated by the chi-square distribution with 4 degrees
of freedom.\footnote{%
   If the two tests had continuously distributed $p$-values, the distribution would be exactly
   chi-square with four degrees of freedom, but if either $p$-value has atoms when
   the null hypothesis is true, it is in general stochastically smaller.
   This follows from a coupling argument along the lines of Theorem~4.12.3 in \citet{grimmett01}.
}

Hence, if, for all $\lambda_1$ and $\lambda_2 = 1- \lambda_1$,
the combined statistic $\chi(\lambda_1, \lambda_2)$ is greater than the 
$1-\alpha$ quantile of the chi-square
distribution with 4 degrees of freedom, the audit can stop.

The calculation of $p_s(\lambda)$ 
uses the procedures discussed in 
sections~\ref{sec:comparisonError} and~\ref{sec:ballotPollError}.
\section{Sampling from subcollections} \label{sec:subcollections}

To audit contests that are contained on only a fraction of the ballots cast
in one or more counties efficiently requires the ability to sample from
just those ballots (or, at least, from a subset of all ballots that contains every such ballot).
Because the CVRs cannot be entirely trusted (otherwise, the audit would be superfluous),
we cannot rely on them to determine which ballots contain a given contest.
However, if we have independent knowledge of the number of ballots that
contain a given contest (e.g., from the SCORE system), then there are methods
that allow the sample to be drawn from ballots whose CVRs contain the contest
and still limit the risk rigorously.
See~\citet{benalohEtal11} and \citet{banuelosStark12} for details.

\section{Batch comparison audits of a tolerable overstatement in votes}
\label{sec:comparisonError}

In this section we expand previous comparison auditing work (already embodied in RLATool) to handle two new requirements.  The first allows the specification of the $\lambda$ parameters discussed in section~\ref{sec:hybrid}. The second handles batch-level auditing.

The first requirement requires that we consider auditing in a single stratum to test whether the overstatement of any margin
(in votes) exceeds some fraction $\lambda$ of the overall margin $V_{w\ell}$ between
reported winner $w$ and reported loser $\ell$.
If the stratum contains all the ballots cast in the contest, then for $\lambda = 1$, this 
would confirm the election outcome.
For stratified audits, we might want to test other values of $\lambda$, as described above.

In Colorado, comparison audits have been ballot-level (i.e., batches consisting of a single
ballot). 
This section also addresses the second requirement by deriving a method for batches of arbitrary size, which might be useful
for Colorado to audit contests that include CVR counties and legacy counties.
We keep the \emph{a priori} error bounds tighter than the ``super-simple'' 
method~\citep{stark10d}.
To keep the notation simpler, we consider only a single contest, but the 
MACRO test statistic \citep{stark09c,stark10d} automatically extends the result to 
auditing $C>1$ contests simultaneously.
The derivation is for plurality contests, including ``vote-for-$k$'' plurality contests.
Majority and super-majority contests are a minor 
modification~\citep{stark08a}.\footnote{%
  So are some forms of preferential and approval voting, such as Borda count, and
  proportional representation contests, such as D'Hondt~\citep{starkTeague14}.
  Changes for IRV/STV are more complicated.
}

\subsection{Notation}
\begin{itemize}
    \item  $\mathcal{W}$: the set of reported winners of the contest
    \item  $\mathcal{L}$: the set of reported losers of the contest
    \item  $N_s$ ballots were cast in all in the stratum. (The contest might not appear on all $N_s$ ballots.)
    \item  $P$ ``batches'' of ballots are in stratum $s$. A batch contains one or more ballots. Every ballot in stratum $s$ is in exactly one batch.
    \item  $n_p$: number of ballots in batch $p$. $N_s = \sum_{p=1}^P n_p$.
    \item  $v_{pi} \in \{0, 1\}$: the reported votes for candidate $i$ in batch $p$
    \item  $a_{pi} \in \{0, 1\}$: actual votes for candidate $i$ in batch $p$. 
                  If the contest does not appear on any ballot in batch $p$, then $a_{pi} = 0$.
                  
    \item  $V_{w\ell,s} \equiv \sum_{p=1}^P (v_{pw} - v_{p\ell})$: 
Reported margin in stratum $s$ of reported winner $w \in \mathcal{W}$ over reported loser 
$\ell \in \mathcal{L}$, in votes.

    \item  $V_{w\ell}$: 
Overall reported margin of reported winner $w \in \mathcal{W}$ over reported loser 
$\ell \in \mathcal{L}$, in votes, for the entire contest (not just stratum $s$)

%
%
    \item  $V$: smallest reported overall margin between any reported winner and reported loser:
$V \equiv \min_{w \in \mathcal{W}, \ell \in \mathcal{L}} V_{w \ell}$

    \item  $A_{w\ell,s} \equiv \sum_{p=1}^P (a_{pw} - a_{p\ell})$: 
actual margin in the stratum of reported winner $w \in \mathcal{W}$ over reported loser 
$\ell \in \mathcal{L}$, in votes

    \item  $A_{w\ell}$: 
actual margin of reported winner $w \in \mathcal{W}$ over reported loser 
$\ell \in \mathcal{L}$, in votes, for the entire contest (not just in stratum $s$)

\end{itemize}

\subsection{Reduction to maximum relative overstatement}
If the contest is entirely contained in stratum $s$, then
the reported winners of the contest are the actual winners if
$$ 
   \min_{w \in \mathcal{W}, \ell \in \mathcal{L}} A_{w\ell,s} > 0.
$$
Here, we address the case that the contest may include a portion outside the stratum.
To combine independent samples in different strata, it is convenient
to be able to test whether the net overstatement error in a stratum exceeds a given threshold.

Instead of testing that condition directly, we will test a condition that is sufficient 
but not necessary for the inequality to hold, to get a computationally simple test that
is still conservative (i.e., the risk is not larger than its nominal value).

For every winner, loser pair $(w, \ell)$, we want to test
whether the overstatement error exceeds some threshold, generally
one tied to the reported margin between $w$ and $\ell$.
For instance, for a simple stratified audit, we might take the threshold to be
$\lambda_s V_{w\ell}$.

We want to test whether
$$
   \sum_{p=1}^P (v_{pw}-a_{pw} - v_{p\ell} + a_{p\ell})/V_{w\ell} \ge \lambda_s.
$$
The maximum of sums is not larger than the sum of the maxima; that is,
$$
\max_{w \in \mathcal{W},  \ell \in \mathcal{L}}
   \sum_{p=1}^P (v_{pw}-a_{pw} - v_{p\ell} + a_{p\ell})/V_{w\ell}
   \le
  \sum_{p=1}^P  \max_{w \in \mathcal{W},  \ell \in \mathcal{L}} 
  (v_{pw}-a_{pw} - v_{p\ell} + a_{p\ell})/V_{w\ell}.
$$

Define 
$$
  e_p \equiv \max_{w \in \mathcal{W} \ell \in \mathcal{L}} 
     (v_{pw}-a_{pw} - v_{p\ell} + a_{p\ell})/V_{w\ell}.
$$
Then no reported margin is overstated by a fraction $\lambda_s$ or more
if 
$$ 
  E \equiv \sum_{p=1}^P e_p < \lambda_s.
$$

Thus if we can reject the hypothesis $E \ge \lambda_s$, we can conclude that
no pairwise margin was overstated by as much as a fraction $\lambda_s$.

Testing whether $E \ge \lambda_s$ would require a very large sample if we knew nothing at
all about $e_p$ without auditing batch $p$: a single large value of $e_p$ could make
$E$ arbitrarily large.
But there is an \emph{a priori} upper bound for $e_p$.
Whatever the reported votes $v_{pi}$ are in batch~$p$, we can find the
potential values of the actual votes $a_{pi}$ that would make the
error $e_p$ largest, because $a_{pi}$ must be between 0 and $n_p$,
the number of ballots in batch~$p$:
$$
    \frac{v_{pw}-a_{pw} - v_{p\ell} + a_{p\ell}}{V_{w\ell}} \le 
    \frac{v_{pw}- 0 - v_{p\ell} + n_p}{V_{w\ell}}.
$$
Hence,
\begin{equation} \label{eq:uDef}
    e_p \le \max_{w \in \mathcal{W}, \ell \in \mathcal{L}} 
    \frac{v_{pw} - v_{p\ell} + n_p}{V_{w\ell}} \equiv u_p.
\end{equation}

Knowing that $e_p \le u_p$ might let us conclude reliably that $E < \lambda_s$
by examining only a small number of batches---depending on the 
values $\{ u_p\}_{p=1}^P$ and on the values of $\{e_p\}$ for the audited batches.

To make inferences about $E$, it is helpful to work with the \emph{taint} 
$t_p \equiv \frac{e_p}{u_p} \le 1$.
Define $U \equiv \sum_{p=1}^P u_p$.
Suppose we draw batches at random with replacement, with probability $u_p/U$
of drawing batch $p$ in each draw, $p = 1, \ldots, P$.
(Since $u_p \ge 0$, these are all positive numbers, and they sum to 1,
so they define a probability distribution on the $P$ batches.)

Let $T_j$ be the value of $t_p$ for the batch $p$ selected in the $j$th draw.
Then $\{T_j\}_{j=1}^n$ are IID, $\mathbb{P} \{T_j \le 1\} = 1$, and
$$
  \mathbb{E} T_1 = \sum_{p=1}^P \frac{u_p}{U} t_p =
  \frac{1}{U}\sum_{p=1}^P u_p \frac{e_p}{u_p} = 
  \frac{1}{U} \sum_{p=1}^P e_p = E/U.
$$
Thus $E = U \mathbb{E} T_1$. 

So, if we have strong evidence that
$\mathbb{E} T_1 < \lambda_s/U$, we have
strong evidence that $E < \lambda_s$.

This approach can be simplified even further by noting that $u_p$ has
a simple upper bound that does not depend on $v_{pi}$.
At worst, the reported result for batch $p$ shows $n_p$ votes for the 
``least-winning'' apparent winner of the contest with the smallest margin, 
but a hand interpretation would show that all $n_p$ ballots in the batch 
had votes for the runner-up in that contest.
Since $V_{w\ell} \ge V$ and $0 \le v_{pi} \le n_p$,
$$ 
    u_p =  \max_{w \in \mathcal{W}, \ell \in \mathcal{L}} 
    \frac{v_{pw} - v_{p\ell} + n_p}{V_{w\ell}}
    \le  \max_{w \in \mathcal{W}, \ell \in \mathcal{L}} 
    \frac{n_p - 0 + n_p}{V_{w\ell}}
    \le \frac{2n_p}{V}.
$$
Thus if we use $2n_p/V$ in lieu of $u_p$, we still get conservative results.
(We also need to re-define $U$ to be the sum of those upper bounds.)
An intermediate, still conservative approach would be to use this upper bound for
batches that consist of a single ballot, but use the sharper bound (\ref{eq:uDef})
when $n_p > 1$.
Regardless, for the new definition of $u_p$ and $U$,
$\{T_j\}_{j=1}^n$ are IID, $\mathbb{P} \{T_j \le 1\} = 1$,
and
$$
  \mathbb{E} T_1 = \sum_{p=1}^P \frac{u_p}{U} t_p =
  \frac{1}{U}\sum_{p=1}^P u_p \frac{e_p}{u_p} = 
  \frac{1}{U} \sum_{p=1}^P e_p = E/U.
$$

So, if we have evidence that $\mathbb{E} T_1 < \lambda_s/U$, we have evidence that 
$E < \lambda_s$.

\subsection{Testing $\mathbb{E} T_1 \ge \lambda_s/U$}

To test whether $\mathbb{E} T_1 < \lambda_s/U$, there are a variety of methods available.
One particularly ``clean'' sequential method is based on Wald's Sequential Probability
Ratio Test (SPRT) (\cite{wald45}).
Harold Kaplan pointed out this method on a website that no longer exists.
A derivation of this ``Kaplan-Wald'' method is given in~\citet[Appendix A]{starkTeague14};
to apply the method here, take $t = \lambda_s$ in their equation~18.

A different sequential method, the Kaplan-Markov method (also due to Harold Kaplan), 
is given in~\citet{stark09b}.

\section{Ballot-polling audits of a tolerable overstatement in votes}
\label{sec:ballotPollError}

\subsection{Conditional tri-hypergeometric test}

We consider a single stratum $s$, containing $N_s$ ballots.
We will sample individual ballots without replacement from stratum $s$.
Of the $N_s$ ballots,
$A_{w,s}$ have a vote for $w$ but not for $\ell$, $A_{\ell,s}$ have a vote for $\ell$ but not for $w$, and $A_{u,s} = N_s - N_{w,s} - N_{\ell,s}$ have votes for both $w$ and $\ell$ or neither $w$ nor $\ell$, including undervotes and invalid ballots.
We might draw a simple random sample of $n$ ballots ($n$ fixed ahead of time), or we might draw 
sequentially without replacement, so the sample size $B$ could be random.
For instance, the rule for determining $B$ could depend on the data.\footnote{%
   Sampling with replacement leads to simpler arithmetic, but is not as efficient.
}

Regardless, we assume that, conditional on the attained sample size $n$, the ballots are a simple random sample of size $n$ from the $N_s$ ballots in the population.
In the sample, $B_w$ ballots contain a vote for $w$ but not $\ell$, with $B_\ell$ and $B_u$ defined analogously.
Then, conditional on $B=n$, the joint distribution of
$(B_w, B_\ell, B_u)$ is tri-hypergeometric:

\begin{equation}
    \mathbb{P}_{A_{w,s}, A_{\ell,s}} \{ B_w = i, B_\ell = j \vert B=n \} = 
     \frac{ {A_{w,s } \choose i}{A_{\ell,s} \choose j}{N_s - A_{w,s} - A_{\ell,s} \choose n-i-j}}{{N_s \choose n}}.
\end{equation}

The test statistic will be the diluted sample margin, $D \equiv (B_w - B_\ell)/B$.
This is the sample difference in the number of ballots for the winner and for the loser, divided by the 
total number of ballots in the sample.
We want to test the compound hypothesis $A_{w,s} - A_{\ell,s} \le c$.
The value of $c$ is inferred from the definition
$\omega_{w\ell,s} \equiv V_{w\ell,s} - A_{w\ell,s} = V_{w,s} - V_{\ell,s} - (A_{w,s} -A_{\ell,s})$.
Thus,
$$
    c = V_{w,s} - V_{\ell,s} - \omega_{w\ell,s} = V_{w\ell,s} - \lambda_s V_{w\ell}.
$$
The alternative is the compound hypothesis 
$A_{w,s} - A_{\ell,s} > c$.\footnote{%
    To use Wald's Sequential Probability Ratio Test, we might pick a simple alternative instead, e.g.,
   $A_{w,s} = V_{w,s}$ and $A_{\ell,s} = V_{\ell,s}$, the reported values, provided 
   $V_{w,s} - V_{\ell,s} > c$.
}
Hence, we will reject for large values of $D$.
Conditional on $B=n$, the event $D = (B_w - B_\ell)/B = d$ is the event $B_w - B_\ell = nd$.

Suppose we observe $D=d$.
The test will condition on the event $B=n$. 
(In contrast, the BRAVO ballot-polling
method~\citep{lindemanEtal12}
conditions only on $B_w+B_\ell = m$.)

The $p$-value of the simple hypothesis that there are $A_{w,s}$ ballots with
a vote for $w$ but not for $\ell$, $A_{\ell,s}$ ballots with a vote for $\ell$ but not for $w$, 
and $N - A_{w,s} - A_{\ell,s}$ ballots with votes for both $w$ and $\ell$ or neither $w$ nor $\ell$ 
(including undervotes and
invalid ballots) is the sum of these probabilities for events when $B_w - B_\ell \geq nd$.
Therefore,

\begin{equation}
   \mathbb{P}_{A_{w,s}, A_{\ell,s}, N_s} \left \{ D \geq d \;\vert\; B = n\right \} = 
   \sum_{\substack{(i, j) :  i, j\ge 0 \\ i-j \geq nd \\ i+j \leq n}} \frac{ {A_{w,s } \choose i}{A_{\ell,s} \choose j}{N_s - A_{w,s} - A_{\ell,s} \choose n-i-j}}{{N_s \choose n}}.
\end{equation}

\subsection{Conditional hypergeometric test}
Another approach is to condition on both the events $B=n$ and $B_w+B_\ell=m$.
We describe the hypothesis test here, but do not advocate for using it.
We found that this approach was inefficient in some simulation experiments.

Given $B=n$, all samples of size $n$ from the ballots are equally likely, by hypothesis.
Hence, in particular, all samples of size $n$ for which $B_w + B_\ell = m$ are equally likely.
There are ${A_{w,s}+A_{\ell,s} \choose m}{N_s - A_{w,s}-A_{\ell,s} \choose n-m}$ such samples.
Among these samples, $B_w$ may take values $i=0, 1, \dots, m$.
For a fixed $i$, there are ${A_{w,s} \choose i}{A_{\ell, s} \choose m-i}{N_s - A_{w,s} - A_{\ell,s} \choose n-m}$
samples with $B_w=i$ and $B_\ell = m-i$.

The factor ${N_s - A_{w,s} - A_{\ell,s} \choose n-m}$ counts the number of ways to sample $n-m$ of the
remaining ballots.
If we divide out this factor, we simply count the number of ways to sample ballots
from the group of ballots for $w$ or for $\ell$.
There are ${A_{w,s}+A_{\ell,s} \choose m}$ equally likely samples of size $m$ from
the ballots with either a vote for $w$ or for $\ell$, but not both, 
and of these samples, ${A_{w,s} \choose i}{A_{\ell, s} \choose m-i}$ contain $i$ ballots with a vote for $w$ but not $\ell$.
Therefore, conditional on $B=n$ and $B_w+B_\ell=m$, the probability that $B_w=i$ is

$$\frac{{A_{w,s} \choose i}{A_{\ell, s} \choose m-i}}{{A_{w,s}+A_{\ell,s} \choose m}}.$$

The $p$-value of the simple hypothesis that there are $A_{w,s}$ ballots with
a vote for $w$ but not for $\ell$, $A_{\ell,s}$ ballots with a vote for $\ell$ but not for $w$, 
and $N - A_{w,s} - A_{\ell,s}$ ballots with votes for both $w$ and $\ell$ or neither $w$ nor $\ell$ 
(including undervotes and
invalid ballots) is the sum of these probabilities for events when $B_w - B_\ell \geq nd$.
This event occurs for $B_w \geq \frac{m+nd}{2}$.
Therefore,

\begin{equation}
   \mathbb{P}_{A_{w,s}, A_{\ell,s}, N_s} \left \{ D \geq d \;\vert\; B = n, B_w+B_\ell = m \right \} = 
   \sum_{i=(m+nd)/2}^{\min\{m, A_{w,s}\}} \frac{{A_{w,s} \choose i}{A_{\ell, s} \choose m-i}}{{A_{w,s}+A_{\ell,s} \choose m}}.
\end{equation}

This conditional $p$-value is thus the tail probability of the hypergeometric distribution
with parameters $A_{w,s}$ ``good'' items, $A_{\ell,s}$ ``bad'' items, and a sample of size $m$.
This calculation is numerically stable and fast; tail probabilities of the hypergeometric distribution are available
and well-tested in all standard statistics software.

\subsection{Maximizing the $p$-value over the null set}

The composite null hypothesis does not specify $A_{w,s}$ or $A_{\ell,s}$ separately, only 
that $A_{w,s} - A_{\ell,s} \le c$ for
some fixed, known $c$.
The (conditional) $p$-value of this composite hypothesis for $D=d$ is the maximum $p$-value for all
values $(A_{w,s}, A_{\ell,s})$ that are possible under the null hypothesis,
\begin{equation}
  \max_{A_{w,s}, A_{\ell,s} \in \{0, 1, \ldots, N \}: A_{w,s} - A_{\ell,s} \le c, A_{w,s} + A_{\ell,s} \le N_s}
   \sum_{\substack{(i, j) :  i, j\ge 0 \\ i-j \geq nd \\ i+j \leq n}} \frac{ {A_{w,s } \choose i}{A_{\ell,s} \choose j}{N_s - A_{w,s} - A_{\ell,s} \choose n-i-j}}{{N_s \choose n}},
\end{equation}
wherever the summand is defined. 
(Equivalently, define ${m \choose k} \equiv 0$ if $k > m$, $k < 0$, or $m \le 0$.)

\subsubsection{Optimizing over the parameter $c$}
The following result enables us to only test hypotheses along the boundary of the null set.

\begin{thm}
Assume that $n < A_{w,s}+A_{\ell,s}$.
Suppose the composite null hypothesis is $N_w - N_\ell \leq c$.
The $p$-value is maximized on the boundary of the null region, i.e. when $N_w - N_\ell = c$.
\end{thm}

\begin{proof}
Without loss of generality, let $c=0$ and assume that $A_{u,s}=N_s - A_{w,s} - A_{\ell,s}$ is fixed.
Let $N_{w\ell, s} \equiv A_{w,s}+A_{\ell,s}$ be the fixed, unknown number of ballots for $w$ or for $\ell$ in stratum $s$.
The $p$-value $p_0$ for the simple hypothesis that $c=0$ is

\begin{equation}
  p_0 = \sum_{\substack{(i, j) :  i, j\ge 0 \\ i-j \geq nd \\ i+j \leq n}} \frac{ {N_{w\ell, s}/2 \choose i}{N_{w\ell, s}/2 \choose j}{A_{u,s} \choose n-i-j}}{{N_s \choose n}} =  \sum_{\substack{(i, j) :  i, j\ge 0 \\ i-j \geq nd \\ i+j \leq n}}T_{ij},
\end{equation}

where $T_{ij}$ is defined as the $(i, j)$ term in the summand and $T_{ij} \equiv 0$ for pairs $(i, j)$ that don't appear in the summation.

Assume that $c>0$ is given.
The $p$-value $p_c$ for this simple hypothesis is
\begin{align*}
p_c &=   \sum_{\substack{(i, j) :  i, j\ge 0 \\ i-j \geq nd \\ i+j \leq n}} \frac{ {(N_{w\ell, s}+c)/2 \choose i}{(N_{w\ell, s}-c)/2 \choose j}{A_{u,s} \choose n-i-j}}{{N_s \choose n}}  \\
   &= \sum_{\substack{(i, j) :  i, j\ge 0 \\ i-j \geq nd \\ i+j \leq n}}T_{ij} \frac{ \frac{N_{w\ell, s}+c}{2}(\frac{N_{w\ell, s}+c}{2}-1)\cdots(\frac{N_{w\ell, s}}{2}+1) (\frac{N_{w\ell, s}-c}{2} -j)\cdots(\frac{N_{w\ell, s}}{2}-1-j) }
   {(\frac{N_{w\ell, s}+c}{2} -i)\cdots(\frac{N_{w\ell, s}}{2}+1-i)(\frac{N_{w\ell, s}-c}{2})\cdots(\frac{N_{w\ell, s}}{2}-1)}.
\end{align*}

Terms in the fraction can be simplified: choose the corresponding pairs in the numerator and denominator.
Fractions of the form $\frac{\frac{N_{w\ell, s}}{2} + a}{\frac{N_{w\ell,s}}{2} + a - i}$ can be expressed as $1 + \frac{i}{\frac{N_{w\ell,s}}{2} + a-i}$.
Fractions of the form $\frac{\frac{N_{w\ell, s}}{2}  - a - j}{\frac{N_{w\ell, s}}{2}  - a}$ can be expressed as $1 - \frac{j}{\frac{N_{w\ell, s}}{2} -a}$.
Thus, the $p$-value can be written as 

\begin{align*}
p_c &= \sum_{\substack{(i, j) :  i, j\ge 0 \\ i-j \geq nd \\ i+j \leq n}}T_{ij} \prod_{a=1}^{c/2} \left(1 + \frac{i}{\frac{N_{w\ell,s}}{2} + a-i}\right)\left(1 - \frac{j}{\frac{N_{w\ell, s}}{2} - a}\right) \\
&> \sum_{\substack{(i, j) :  i, j\ge 0 \\ i-j \geq nd \\ i+j \leq n}} T_{ij} \left[ \left(1 + \frac{i}{\frac{N_{w\ell,s}+c}{2} -i}\right)\left(1 - \frac{j}{\frac{N_{w\ell, s}}{2}+1}\right) \right]^{c/2} \\
&= \sum_{\substack{(i, j) :  i, j\ge 0 \\ i-j \geq nd \\ i+j \leq n}} T_{ij} \left[ 1 + \frac{\frac{N_{w\ell,s}+c}{2}j + \frac{N_{w\ell,s}}{2}i + i}{(\frac{N_{w\ell,s}+c}{2}-i)(\frac{N_{w\ell,s}}{2}+1)}\right]^{c/2} \\
&> \sum_{\substack{(i, j) :  i, j\ge 0 \\ i-j \geq nd \\ i+j \leq n}} T_{ij}\\
&= p_0
\end{align*}

The last inequality follows from the fact that $i$ and $j$ are nonnegative, and 
that $i < \frac{N_{w\ell,s}+c}{2}$ -- it is a possible outcome under the null hypothesis.

\end{proof}

\subsubsection{Optimizing over the parameter $A_{w,s}$}

We have shown empirically (but do not prove) that this tail probability, as a function of $A_{w,s}$,
has a unique maximum at one of the endpoints when $A_{w,s}$ is either as small or as large as possible,
given $N$, $c$, and the observed sample values $B_w$ and $B_\ell$.
If the empirical result is true, then finding the maximum is trivial;
otherwise, it is a trivial one-dimensional optimization problem to compute the unconditional $p$-value.

\subsection{Conditional testing}
If the conditional tests are always conducted at significance level $\alpha$ or less, i.e., so that
$\mathbb{P} \{\mbox{Type I error} | B = n\} \le \alpha$, then the
overall procedure has significance level $\alpha$ or less:
\begin{eqnarray}
    \mathbb{P} \{\mbox{Type I error}\} &=& \sum_{n=0}^N \{\mbox{Type I error} |  B = n\} \mathbb{P} \{ B = n \} \nonumber \\
       & \le & \sum_{n=0}^N \alpha \mathbb{P} \{  B = n \}  =  \alpha.
\end{eqnarray}

In particular, this implies that our conditional hypergeometric test will have the correct risk limit unconditionally.

\section{Recommendations} \label{sec:recommendations}

We have outlined several methods Colorado might use to audit cross-jurisdictional contests
that include CVR counties and no-CVR counties.
We expect that stratified ``hybrid'' audits will be the most palatable,
given the constraints on time for software development and the logistics
of the audit itself, because the workflow for counties would be the same
as it was in November, 2017.

What would change is the risk calculation ``behind the scene,'' including the
algorithms used to decide when the audit can stop.
Those algorithms could be implemented in software external to RLATool.
The minimal modification to RLATool that would be required to conduct
a hybrid audit is to allow the sample size from each county to be controlled externally,
e.g. by uploading a parameter file once per round,
rather than using a formula that is based on the margin within that county alone.
The parameter file would be generated by external software that does the
audit calculations described here based on the detailed
audit progress and discrepancy data available from RLATools' rla\_export command.

To conduct a hybrid audit, one must choose two numbers
in addition to the risk limit $\alpha$:
\begin{itemize}
   \item one stratum-wise risk limit, $\alpha_1$ 
(the other, $\alpha_2$, is determined from $\alpha_1$ and the overall risk limit, $\alpha$)
   \item the tradeoff (allocation) of the tolerable overstatement between strata, $\lambda_1$
(the value of $\lambda_2$ is $1-\lambda_1$)
\end{itemize}
Those parameters can be chosen essentially arbitrarily (provided $\alpha_1 \le \alpha$)
and the audit will still be risk-limiting;
however, they can be optimized to reduce the expected workload under various
assumptions about tabulation errors in the two strata.
(Software that can be used to run scenarios is available at 
\url{https://www.github.com/pbstark/CORLA18}; see below.)

In either stratum, increasing the risk limit or increasing the tolerable overstatement will decrease
the required sample size from that stratum (assuming that the actual overstatement in
that stratum is less than its allowable overstatement).
The relative change in sample size as the risk limit changes scales similarly in the two strata, 
because the risk limit 
enters both ballot-level comparisons and ballot-polling the same way: as the logarithm.
However, the relative change in sample sizes as the tolerable overstatement changes
scales quite differently in the two strata: linearly in the ballot-level comparison stratum,
but quadratically in the ballot-polling stratum.
Hence, the workload is not as sensitive to how the risk limit is allocated across strata
as it is to how the tolerable overstatement is allocated.

\subsection{Software and examples}
Examples of stratified hybrid audits are in Jupyter notebooks available
at \url{https://www.github.com/pbstark/CORLA18}.
The first two examples are contained in a single notebook, ``hybrid-audit-example-1''.
The first example is a hypothetical medium-sized election with a total of 
$110,000$ votes and a diluted margin of $1.8\%$.
9.1\% of the ballots come from no-CVR counties. The risk limit is 10\%.
If the audit in the CVR stratum found no errors and the allowable overstatement error was 30\% of the margin, 
it would terminate after examining 1,213 ballots.
In over 90\% of 10,000 simulations, an audit of 250 ballots from the no-CVR stratum
would have sufficed to confirm that the overstatement error in that stratum
did not exceed its allocation, 70\% of the margin.
A sample of 450 ballots was sufficient to stop the audit in 99\% of simulations.
As always, $\lambda_1$ could be adjusted to
rebalance the expected workload between strata, perhaps taking into account the expected
workload for audits of countywide or intra-county contests, so as to minimize (or quite possibly
eliminate) any additional burden imposed by the stratified audit.

If a CVR were available for all counties and we could have run a ballot-level comparison audit for the entire contest, 
rather than stratifying, an audit with risk limit 10\% that found no errors would have concluded after examining just 263 ballots.
The efficiency gained comes from two sources.
First, ballot-level comparison audits are substantially more efficient than ballot-polling audits.
Second, the hybrid audit requires dividing the margin and risk limit between two strata.
This results in both strata using smaller risk limits.
In order to keep the workload low, it is necessary to allocate a disproportionately high fraction of the margin
to the no-CVR stratum;
the CVR stratum must increase its workload to compensate.

Another method discussed in Section~\ref{sec:two-vote-over} is to perform a ballot-level comparison audit statewide,
but to treat any ballot sampled from the no-CVR county as showing a two-vote overstatement.
In this example, this worst-case method would lead to a full hand count.
However, the situation may be more optimistic for Colorado:
if only 1.2\% of ballots came from the no-CVR stratum and the overall margin were in fact 10,000 votes,
then this method would require checking 430 ballots.

The second example is a hypothetical large statewide election with a total of 
2~million ballots and a diluted margin of nearly $20\%$.  The risk limit is 5\%.
If the audit in the CVR stratum found no errors and the allowable overstatement error was 10\% of the margin, 
it would terminate after examining 50 ballots.
In over 90\% of 10,000 simulations, an audit of 50 ballots from the no-CVR stratum
would have sufficed to confirm that the overstatement error in that stratum
did not exceed its allocation, 90\% of the margin.
A sample of 100 ballots was sufficient to stop the audit in 99\% of simulations.
If a CVR were available for all counties and we could have run a ballot-level comparison audit for the entire contest, 
rather than stratifying, an audit with risk limit 5\% that found no errors would have concluded after examining just 24 ballots.

A second notebook, ``hybrid-audit-example-2,'' illustrates the workflow for conducting a hybrid audit of this kind.
The example election has a total of 2~million ballots.
The reported margin is just over $1\%$, but in reality the vote totals for the reported winner
and reported loser are identical in both strata.  The risk limit is 5\%.
The example illustrates two scenarios.
In the first scenario, the audit in the CVR stratum escalates to a full hand count and the allowable overstatement
in the no-CVR stratum must be adjusted.
Using the new allowable overstatement in the no-CVR stratum makes it impossible to terminate the audit,
even for samples as large as 5\% of the ballots.
In the second scenario, the audit in the no-CVR stratum terminates with a sample of 500 ballots.
However, the audit in the CVR stratum will still lead to a full hand count and the audit in the no-CVR
stratum must be redone using the adjusted allowable overstatement, putting us back in the first scenario.
In both cases, the audit leads to a full recount of all the ballots.

These notebooks can be modified and run with different contest sizes, margins, risk limits, and allocations of allowable error, in order to
estimate the workload of different scenarios.

\bibliography{./pbsBib}

\begin{thebibliography}{13}
\providecommand{\natexlab}[1]{#1}
\providecommand{\url}[1]{\texttt{#1}}
\expandafter\ifx\csname urlstyle\endcsname\relax
  \providecommand{\doi}[1]{doi: #1}\else
  \providecommand{\doi}{doi: \begingroup \urlstyle{rm}\Url}\fi

\bibitem[Ba\~{n}uelos and Stark(2012)]{banuelosStark12}
J.H. Ba\~{n}uelos and P.B. Stark.
\newblock Limiting risk by turning manifest phantoms into evil zombies.
\newblock Technical report, {arXiv.org}, 2012.
\newblock URL \url{http://arxiv.org/abs/1207.3413}.
\newblock Retrieved 17 July 2012.

\bibitem[Benaloh et~al.(2011)Benaloh, Jones, Lazarus, Lindeman, and
  Stark]{benalohEtal11}
J.~Benaloh, D.~Jones, E.~Lazarus, M.~Lindeman, and P.B. Stark.
\newblock {SOBA}: Secrecy-preserving observable ballot-level audits.
\newblock In \emph{Proceedings of the 2011 Electronic Voting Technology
  Workshop / Workshop on Trustworthy Elections ({EVT/WOTE} '11)}. {USENIX},
  2011.
\newblock URL \url{http://statistics.berkeley.edu/~stark/Preprints/soba11.pdf}.

\bibitem[Grimmett and Stirzaker(2001)]{grimmett01}
Geoffrey~R. Grimmett and David~R. Stirzaker.
\newblock \emph{Probability and Random Processes}.
\newblock {Oxford University Press}, August 2001.
\newblock ISBN 0198572220.
\newblock URL
  \url{http://www.amazon.ca/exec/obidos/redirect?tag=citeulike09-20\&amp;path=ASIN/0198572220}.

\bibitem[Lindeman et~al.(2012)Lindeman, Stark, and Yates]{lindemanEtal12}
M.~Lindeman, P.B. Stark, and V.~Yates.
\newblock {BRAVO}: Ballot-polling risk-limiting audits to verify outcomes.
\newblock In \emph{Proceedings of the 2011 Electronic Voting Technology
  Workshop / Workshop on Trustworthy Elections (EVT/WOTE '11)}. {USENIX}, to
  appear 2012.

\bibitem[Lindeman and Stark(2012)]{lindemanStark12}
Mark Lindeman and Philip~B. Stark.
\newblock A gentle introduction to risk-limiting audits.
\newblock \emph{IEEE Security and Privacy}, 10:\penalty0 42--49, 2012.

\bibitem[Rivest(2018)]{Rivest-2018-bayesian-tabulation-audits}
Ronald~L. Rivest.
\newblock Bayesian tabulation audits: Explained and extended, January 1, 2018.
\newblock URL \url{https://arxiv.org/abs/1801.00528}.

\bibitem[Stark(2008)]{stark08a}
P.B. Stark.
\newblock Conservative statistical post-election audits.
\newblock \emph{Ann. Appl. Stat.}, 2:\penalty0 550--581, 2008.
\newblock URL \url{http://arxiv.org/abs/0807.4005}.

\bibitem[Stark(2009{\natexlab{a}})]{stark09b}
P.B. Stark.
\newblock Risk-limiting post-election audits: {$P$}-values from common
  probability inequalities.
\newblock \emph{IEEE Transactions on Information Forensics and Security},
  4:\penalty0 1005--1014, 2009{\natexlab{a}}.

\bibitem[Stark(2009{\natexlab{b}})]{stark09c}
P.B. Stark.
\newblock Auditing a collection of races simultaneously.
\newblock Technical report, {arXiv.org}, 2009{\natexlab{b}}.
\newblock URL \url{http://arxiv.org/abs/0905.1422v1}.

\bibitem[Stark(2010)]{stark10d}
P.B. Stark.
\newblock Super-simple simultaneous single-ballot risk-limiting audits.
\newblock In \emph{Proceedings of the 2010 Electronic Voting Technology
  Workshop / Workshop on Trustworthy Elections ({EVT/WOTE} '10)}. {USENIX},
  2010.
\newblock URL
  \url{http://www.usenix.org/events/evtwote10/tech/full_papers/Stark.pdf}.

\bibitem[Stark and Teague(2014)]{starkTeague14}
Philip~B. Stark and Vanessa Teague.
\newblock Verifiable european elections: Risk-limiting audits for d'hondt and
  its relatives.
\newblock \emph{JETS: USENIX Journal of Election Technology and Systems}, 3.1,
  2014.
\newblock URL \url{https://www.usenix.org/jets/issues/0301/stark}.

\bibitem[Stark and Wagner(2012)]{starkWagner12}
Philip~B. Stark and David~A. Wagner.
\newblock Evidence-based elections.
\newblock \emph{IEEE Security and Privacy}, 10:\penalty0 33--41, 2012.

\bibitem[Wald(1945)]{wald45}
A.~Wald.
\newblock Sequential tests of statistical hypotheses.
\newblock \emph{Ann. Math. Stat.}, 16:\penalty0 117--186, 1945.

\end{thebibliography}

\end{document}